\documentclass{llncs}
\usepackage{tikz}
\usepackage{proof}
\usepackage{color}
\usepackage[vlined,figure]{algorithm2e}
\usepackage{pifont} 
\usepackage{booktabs}
\usepackage{amsmath}

\newif\ifproofs
\proofstrue 

\usetikzlibrary{arrows,positioning,shapes}
\def\checkmark{\tikz\fill[scale=0.3](0,.35) -- (.25,0) -- (1,.7) -- (.25,.15) -- cycle;} 
\newcommand\BOX[1]{\boxed{\mbox{$#1$}}}

\newcommand{\xmark}{\ding{55}}
\definecolor{dkgreen}{rgb}{0,0.5,0}
\definecolor{dkred}{rgb}{0.5,0,0}
\definecolor{dkgray}{rgb}{0.3,0.3,0.3}
\usepackage{listings}
\newcommand\lt[1]{{\lstinline@#1@}} 

\begin{document}
\pagestyle{headings}  

\title{Connecting Program Synthesis and Reachability: {\large Automatic Program Repair using Test-Input Generation}}
\author{ThanhVu Nguyen\inst{1} \and Westley Weimer\inst{2} \and
Deepak Kapur\inst{3} \and Stephanie Forrest\inst{3}}
\institute{University of Nebraska, Lincoln NE, USA, \email{tnguyen@cse.unl.edu}\\
\and University of Virginia, Charlottesville VA, USA, \email{weimer@virginia.edu}\\
\and University of New Mexico, Albuquerque NM, USA, \email{ \{kapur,forrest\}@cs.unm.edu}}

\maketitle              

\begin{abstract}

We prove that certain formulations of program synthesis and reachability are equivalent.
Specifically, our constructive proof shows the reductions between the
template-based synthesis problem, which generates a program in a
pre-specified form, and the reachability problem, which decides the reachability of a program location.
This establishes a link between the two research fields and allows for the transfer of techniques and results between them.

To demonstrate the equivalence, we develop a program repair prototype using reachability tools. 
We transform a buggy program and its required specification into a specific program containing a location reachable only when the original program can be repaired, and then apply an off-the-shelf test-input generation tool on the transformed program to find test values to reach the desired location.
Those test values correspond to repairs for the original program.
Preliminary results suggest that our approach compares favorably to other repair methods.

\end{abstract}

\section{Introduction}

Synthesis is the task of generating a program that meets a required
specification.  Verification is the task of validating program
correctness with respect to a given specification.  Both are
long-standing problems in computer science, although there has been
extensive work on program verification and comparatively less on
program synthesis until recently.  Over the past several years,
certain verification techniques have been adopted to create programs,
e.g., applying symbolic execution to synthesize program
repairs~\cite{konighoferautomated,konighofer2013repair,mechtaev2016angelix,nguyen2013semfix}, 
suggesting the possibility that these two problems may be ``two sides of the same coin''.
Finding and formalizing this equivalence is valuable in both theory and
practice: it allows comparisons between the complexities and
underlying structures of the two problems, and it raises the
possibility of additional cross-fertilization between two fields that are usually
treated separately (e.g., it might enable approximations designed to solve one problem to be applied directly to the other).

This paper establishes a formal connection between certain formulations of program synthesis and verification.
We focus on the \emph{template-based synthesis} problem, which generates missing code for partially completed programs, and we view verification as a \emph{reachability} problem, which checks if a program can reach an undesirable state.
We then constructively prove that template-based synthesis and reachability are \emph{equivalent}.
We reduce a template-based synthesis problem, which consists of a program with parameterized templates to be synthesized and a test suite specification, to a program consisting of a specific location that is reachable only when those templates can be instantiated such that the program meets the given specification.
To reduce reachability to synthesis, we transform a reachability instance consisting of a program and a given location into a synthesis instance that can be solved only when the location in the original problem is reachable. 
Thus, reachability solvers can be applied to synthesize code, and
conversely, synthesis tools can be used to determine reachability.

To demonstrate the equivalence, we use the reduction to develop a new automatic program repair technique using an existing test-input generation tool.
We view \emph{program repair} as a special case of template-based synthesis in which ``patch'' code is generated so that it behaves correctly. 
We present a prototype tool called CETI that automatically repairs C
programs that violate test-suite specifications.
Given a test suite and a program failing at least one test in that suite, CETI first applies fault localization to obtain a list of ranked suspicious statements from the buggy program. 
For each suspicious statement, CETI transforms the buggy program and the information from its test suite into a program reachability instance. 
The reachability instance is a new program containing a special {\tt if} branch, whose {\tt then} branch is reachable only when the original program can be repaired by modifying the considered statement. 
By construction, any input value that allows the special location to be reached can map directly to a repair template instantiation that fixes the bug. 
To find a repair, CETI invokes an off-the-shelf automatic test-input
generation tool on the transformed code to find test values that can
reach the special branch location.  
These values correspond to changes that, when applied to the original
program, cause it to pass the given test suite.  This procedure is guaranteed to be
sound, but it is not necessarily complete.  That is, there  may be
bugs that the procedure cannot find repairs for, but all proposed repairs are guaranteed to be correct with respect to the given
test suite.  We evaluated CETI on the {\tt Tcas} program~\cite{SIR}, which has 41 seeded defects, and found that it repaired over 60\%, which compares
favorably with other state-of-the-art automated bug repair approaches.

To summarize, the main contributions of the paper include:
\begin{itemize}
\item \emph{Equivalence Theorem}: We constructively prove that the
  problems of template-based program synthesis and reachability in
  program verification are equivalent.
  Even though these two problems are shown to be undecidable in general, the constructions allow heuristics solving one problem to be applied to the other.

\item \emph{Automatic Program Repair}: We present a new automatic
  program repair technique, which leverages the construction.  The technique
reduces the task of synthesizing program repairs to a
reachability problem, where the results produced by a test-input
generation tool correspond to a patch that repairs the original program.

\item \emph{Implementation and Evaluation}: We implement the repair
  algorithm in a prototype tool that automatically repairs C programs,
  and we evaluate it on a benchmark that has been targeted by multiple program repair algorithms. 
\end{itemize}

\section{Motivating Example}\label{sec:motiv}

\begin{figure}[t]
\begin{minipage}{0.55\linewidth}
\begin{lstlisting}[numberstyle=\scriptsize,xleftmargin=0.00cm,basicstyle=\scriptsize,language=C]
int is_upward(int in,int up,int down){
  int bias, r;
  if (in) 
    bias = down; //fix: bias = up + 100
  else
    bias = up;
  if (bias > down)
    r = 1;
  else
    r = 0;
  return r;
}
\end{lstlisting}
\end{minipage}
\begin{minipage}{0.44\linewidth}
{\scriptsize
  \begin{tabular}{c | r r r | c c |c}
    &\multicolumn{3}{c}{{\bf Inputs}}&\multicolumn{2}{c}{{\bf Output}}\\
    {\bf Test}&in&up&down &  expected &observed &{\bf Passed?}\\
    \midrule
    1 & 1&  0&100  & 0& 0 & \checkmark \\
    2 & 1& 11&110  & 1& 0 & \xmark \\
    3 & 0&100& 50  & 1& 1 & \checkmark \\
    4 & 1&-20& 60  & 1& 0 & \xmark \\
    5 & 0&  0& 10  & 0& 0 & \checkmark \\
    6 & 0&  0& -10 & 1& 1 & \checkmark \\
  \end{tabular}
}
\end{minipage}
\caption{Example buggy program and test suite.  CETI suggests replacing line 4 with the statement {\tt bias = up + 100;} to repair the bug.}
\label{fig1}
\end{figure}

We give a concrete example of how the reduction from template-based synthesis to reachability can be used to repair a buggy program.
Consider the buggy code shown in Figure~\ref{fig1}, a function excerpted from a traffic collision avoidance system~\cite{SIR}.
The intended behavior of the function can be precisely described as:
{\tt is\_upward(in,up,down) = in*100 + up > down}.
The table in Figure~\ref{fig1} gives a test suite describing the intended
behavior. The buggy program fails two of the tests, which we propose
to repair by synthesizing a patch.

We solve this synthesis problem by restricting ourselves to generating patches under predefined templates, e.g., synthesizing expressions involving program variables and unknown parameters, and then transforming this template-based synthesis problem into a reachability problem instance.
In this approach, a template such as
\[
\BOX{c_0}+\BOX{c_1}v_1+\BOX{c_2}v_2
\]
is a linear combination\footnote{More general templates (e.g., nonlinear polynomials) are also possible as shown in Section~\ref{sec:equiv}.}
 of program variables $v_i$ and unknown template parameters $\BOX{c_i}$. 
For clarity, we often denote template parameters with a box to
distinguish them from normal program elements. 
This template can be instantiated to yield concrete expressions such as
$200 + 3v_1 + 4v_2$ via $c_0=200, c_1=3, c_2=4$. 
To repair Line 4 of Figure~\ref{fig1}, 
({\tt bias = down;}) with a linear template, we would replace
Line 4 with: 

{\tt bias = \BOX{c_0} +\BOX{c_1}*bias +\BOX{c_2}*in +\BOX{c_3}*up +\BOX{c_4}*down;}

\noindent where {\tt bias}, {\tt in}, {\tt up}, and {\tt down} are the
variables in scope at Line 4 and the value of each $c_i$ must be found. 
We propose to find them by constructing a special program reachability
instance and then solving that instance. 

The construction transforms the program, its test suite (Figure~\ref{fig1}), and
the template statement into a reachability instance
consisting of a program and target location. 
The first key idea is to derive a new program containing the template
code with the template parameters \BOX{c_i} represented
explicitly as program variables $c_i$. 
This program defines the reachability instance, which must assign values to each $c_i$. 
The second key idea is that each test case is explicitly represented as a conditional expression. Recall that we
seek a single synthesis solution (one set of values for $c_i$) that
respects all tests.  Each test is encoded as a conditional expression (a
reachability constraint), and we take their conjunction, being careful to refer
to the same $c_i$ variables in each expression.  In the example, we must
find one repair that satisfies all six tests, not six separate repairs that
each satisfy only one test. 

The new program, shown in Figure~\ref{fig2}, contains a function {\tt is\_upward$_P$} that resembles the function {\tt is\_upward}
in the original code but with Line 4 replaced by the template statement
with each reference to a template parameter replaced by a reference to
the corresponding new externally-defined program variable.
The program also contains a starting function {\tt main},
which encodes the inputs and expected outputs from the
given test suite as the guards to the conditional statement leading to the
target location $L$.  
Intuitively, the reachability problem
instance asks if we can find values for each $c_i$ that allow control flow
to reach location $L$, which is only reachable iff all tests are satisfied.

\begin{figure}[t]
 \begin{lstlisting}[numbers=none,xleftmargin=0.0cm,basicstyle=\scriptsize,multicols=2]
int c$_0$,c$_1$,c$_2$,c$_3$,c$_4$; //global inputs

int is_upward$_P$(int in,int up,int down){
  int bias, r;
  if (in) 
    bias = 
     c$_0$+c$_1$*bias+c$_2$*in+c$_3$*up+c$_4$*down;
  else 
    bias = up;
  if (bias > down) 
     r = 1;
  else 
     r = 0;
  return r;
}

int main() { 
   if(is_upward$_P$(1,0,100) == 0 &&
      is_upward$_P$(1,11,110) == 1 &&
      is_upward$_P$(0,100,50) == 1 &&
      is_upward$_P$(1,-20,60) == 1 &&
      is_upward$_P$(0,0,10) == 0 &&
      is_upward$_P$(0,0,-10) == 1){
     [L]                                
   }
   return 0;
}
\end{lstlisting}
\caption{The reachability problem instance derived from the
buggy program and test suite in Figure~\ref{fig1}. Location $L$ is
reachable with values such as $c_0=100,c_1=0,c_2=0,c_3=1,c_4=0$. These values
suggest using the statement {\tt bias = 100 + up;} at Line 4 in the
buggy program.}
\label{fig2}
\end{figure}  

This reachability instance can be given as input to any
off-the-self test-input generation tool. 
Here, we use KLEE~\cite{cadar2008klee} to find value for
each $c_i$. 
KLEE determines that the values $c_0 = 100, c_1 =0, c_2 =0, c_3=1, c_4=0$ allow control flow to reach location $L$.
Finally, we map this solution back to the original program repair
problem by applying the $c_i$ values to the template

{\tt bias = \BOX{c_0} +\BOX{c_1}*bias +\BOX{c_2}*in +\BOX{c_3}*up +\BOX{c_4}*down;}

\noindent generating the statement:

{\tt  bias = 100 + 0*bias + 0*in + 1*up + 0*down;}

\noindent which reduces to {\tt bias = 100 + up}.
Replacing the statement {\tt bias = down} in the original program
with the new statement {\tt bias = 100 + up} produces a program that passes all of the test cases.

To summarize, a specific question (i.e., can the bug be
repaired by applying template $X$ to line $Y$ of program $P$ while satisfying
test suite $T$?)  is reduced to a single reachability instance, solvable using a reachability tool such as a test-input generator. 
This reduction is formally established in the next section.

\section{Connecting Program Synthesis and Reachability}\label{sectlpsr}
We establish the connection between the template-based formulation of program synthesis and the reachability problem in program verification.
We first review these problems and then show their equivalence.

\subsection{Preliminaries}
We consider standard imperative programs in a language like C.
The base language includes usual program constructs such as assignments, conditionals, loops, and functions.
A function takes as input a (potentially empty) tuple of values and returns an output value.
A function can call other functions, including itself. 
For simplicity,  we equate
a program $P$ with its finite set of functions, including a special starting function $\text{main}_P$.
For brevity, we write $P(c_i,\dots,c_n)=y$ to denote that evaluating
the function $\text{main}_P \in P$ on the input tuple
$(c_i,\dots,c_n)$ results in the value $y$.  Program or function
semantics are specified by a test suite consisting of a finite set of input/output pairs.
When possible, we use $c_i$ for concrete input values and 
$v_i$ for formal parameters or variable names.

To simplify the presentation, we assume that the
language also supports exceptions, 
admitting non-local control flow by raising and catching exceptions as in modern programming languages such as C++ and Java. 
We discuss how to remove this assumption in Section~\ref{r2s}.

\subsubsection*{Template-based Program Synthesis.} 
\emph{Program synthesis} aims to automatically generate program code to meet a required specification. 
The problem of synthesizing a complete program is generally undecidable~\cite{srivastava2010satisfiability}, so many practical synthesis techniques operate on partially-complete programs, filling in well-structured gaps~\cite{solar2005programming,Srivastava:2010:PVP:1707801.1706337,lezama2008program,alchemist,alur2015syntax,srivastava2013template}. 
These techniques synthesize programs from specific grammars,
forms, or templates and do not generate arbitrary code.
A synthesis \emph{template} expresses the shape of program constructs,
but includes holes (sometimes called template parameters), as
illustrated in the previous section.  We refer to a program containing
such templates as a \emph{template program} and 
extend the base language to include
a finite, fixed set of template parameters \boxed{c_i} as shown earlier.
Using the notation of contextual operational semantics, we write $P[c_0, \dots, c_n]$ to denote the result of instantiating the template program $P$ with template parameter values $c_0 \dots c_n$. 
To find values for the parameters in a template program, many techniques (e.g.,~\cite{solar2005programming,Srivastava:2010:PVP:1707801.1706337,alur2015syntax,srivastava2013template})
 encode the program and its specification as a logical formula (e.g., using axiomatic semantics) and use a constraint solver such as SAT or SMT to find values for the parameters $c_i$ that satisfy the formula.
Instantiating the templates with those values produces a complete program that adheres to the required specification.

\begin{definition}\label{defps}
  {\bf Template-based Program Synthesis Problem.} 
  Given a template program $Q$ with a finite set of template parameters 
  $S=\{\boxed{c_1}, \dots, \boxed{c_n}\}$ and a finite test suite of
  input/output 
  pairs $T=\{ (i_1,o_1), \dots, (i_m,o_m) \}$, do there exist
  parameter values $c_i$ such that $\forall (i,o) \in T ~.~ (Q[c_1, \dots, c_n])(i) = o$?
\end{definition}

\noindent 
For example, the program in Figure~\ref{fig1} with Line 4 replaced by  
{\tt bias = \BOX{c_0} +\BOX{c_1}*bias +\BOX{c_2}*in +\BOX{c_3}*up +\BOX{c_4}*down} 
is an instance of template-based synthesis.
This program passes its test suite given in Figure~\ref{fig1} using the solution $\{c_0=100,c_1=1,c_2=0,c_3=1,c_4=0\}$. 
The decision formulation of the problem asks if satisfying values $c_1 \dots c_n$ exist; in this presentation we require that witnesses be produced. 

\subsubsection{Program Reachability.}

\emph{Program reachability} is a classic problem which asks if a particular program state or location can be observed at run-time.
It is not decidable in general, because it can encode the halting problem
(cf. Rice's Theorem~\cite{rice1953classes}).
However, reachability remains a popular and well-studied verification problem in practice. 
In model checking~\cite{clarke1999model}, for example, reachability is used
to determine whether program states representing undesirable behaviors
could occur in practice.  Another application area is 
test-input generation~\cite{cadar2013symbolic}, which aims to produce test values to explore all reachable program locations.

\begin{definition}\label{defpr}
{\bf Program Reachability Problem.}
Given a program $P$, set of program variables $\{ x_1,\dots,x_n \}$ and 
target location $L$, do there exist input values $c_i$ such that the
execution of $P$ with $x_i$ initialized to $c_i$ 
reaches $L$ in a finite number of steps?
\end{definition}

\begin{figure}[t]
\centering
\begin{minipage}{0.38\linewidth}
\begin{lstlisting}[numbers=none,xleftmargin=0.5cm,basicstyle=\scriptsize,showlines]
//global inputs
int x, y; 

int P(){
  if (2 * x == y) 
    if (x > y + 10) 
      [L]
  
  return 0;
} 

\end{lstlisting}
\caption{An instance of program reachability. Program $P$ reaches location $L$ using the solution $\{x=-20,y=-40\}$.}
\label{test-input}
\end{minipage}
\hspace{\fill}
\begin{minipage}{0.58\linewidth}
\begin{lstlisting}[numbers=none,xleftmargin=0.0cm,basicstyle=\scriptsize,showlines,multicols=2]

int P$_Q$() {
  if (2*$\boxed{x}$ == $\boxed{y}$)
    if($\boxed{x}$ > $\boxed{y}$+10)
      //loc L in P
      raise REACHED;

  return 0;
} 

int main$_Q$() {
  //synthesize x, y
  int x = c$_x$; 
  int y = c$_y$;
  try 
    P$_Q$();
  catch (REACHED)
    return 1; 

  return 0;
} 
\end{lstlisting}
\caption{Reducing the reachability example in Figure~\ref{test-input} to a
template-based synthesis program (i.e., synthesize assignments to
$c_x$ and $c_y$). The test suite of the reduced synthesis program is $Q() = 1$.}
\label{test-input-syn}
\end{minipage}
\end{figure}

\noindent
For example, the program in Figure~\ref{test-input} has a reachable location $L$ using the solution $\{x=-20,y=-40\}$. 
Similar to the synthesis problem, the decision problem formulation of reachability merely asks if the input values $c_1,\dots,c_n$ exist; in this presentation we 
require witnesses be produced. 

\subsection{Reducing Synthesis to Reachability}\label{s2r}

We present the constructive reduction from synthesis to
reachability.  The key to the reduction is a particular ``gadget'',
which constructs a reachability instance that can be satisfied iff the synthesis
problem can be solved. 

\paragraph{Reduction:} Let $Q$ be a template program with a set of template parameters $S=\{\boxed{c_1}, \dots, \boxed{c_n}\}$ and a set of finite tests $T=\{ (i_1,o_1), \dots \}$.  
We construct $\mathsf{GadgetS2R}(Q,S,T)$, which returns a new program $P$ (the constructed reachability instance) with a special location $L$, as follows: 

\begin{enumerate}

\item For every template parameter $\boxed{c_i}$, add a fresh global
variable $v_i$. A solution to this reachability instance is an assignment
of concrete values $c_i$ to the variables $v_i$. 

\item For every function $q\in Q$, define a similar function $q_P \in P$. 
The body of $q_P$ is the same as $q$, 
but with every reference to a template parameter $\boxed{c_i}$ replaced
with a reference to the corresponding new variable $v_i$.  

\item $P$ also contains a starting function main$_P$ that encodes the specification information from the test suite $T$ as a conjunctive expression $e$:
\[
e = \bigwedge_{(i,o) \, \in\, T} \text{main}_Q{_P}(i) = o
\]

where main$_Q{_P}$ is a function in $P$ corresponding to the starting function main$_Q$ in $Q$.
In addition, the body of main$_P$ is one conditional statement
leading to a fresh target location $L$ if and only if $e$ is true. 
Thus, main$_P$ has the form
    
\begin{lstlisting}[numbers=none,xleftmargin=1.cm,basicstyle=\scriptsize]
int main$_P$() {
  if ($e$)
    [L]
} 
\end{lstlisting}

\item 
The derived program $P$ consists of the declaration of the new variables 
(Step 1), the functions $q_P$'s (Step 2), and the starting function main$_P$ (Step 3).
\end{enumerate}

\paragraph{Example:} Figure~\ref{fig2} illustrates the reduction using the example
from Figure~\ref{fig1}.  
The resulting reachability program can arrive at location $L$ using
the input $\{c_0=100,c_1=0,c_2=0,c_3=1,c_4=0\}$, which corresponds to a solution.

\paragraph{Reduction Correctness and Complexity:} The correctness of $\mathsf{GadgetS2R}$, which transforms synthesis to reachability, relies on two key invariants\footnote{The full proof is given in the Appendix of~\cite{equiv_unl_tr}.}. 
First, function calls in the derived program $P$ 
have the same behavior as template functions in the original program $Q$.
Second, location $L$ is reachable if and only if values $c_i$ can be
assigned to variables $v_i$ such that $Q$ passes all of the tests. 

The complexity of $\mathsf{GadgetS2R}$ is \emph{linear} in both the program size and number of test cases of the input instance $Q,S,T$.
The constructed program $P$ consists of all functions in $Q$ (with $|S|$
extra variables) and a starting function main$_P$ with an expression encoding the test suite $T$.

This reduction directly leads to the main result for this direction of the equivalence:
\begin{theorem}\label{thsr-simple}
The template-based synthesis problem in Definition~\ref{defps} is reducible to the reachability problem in Definition~\ref{defpr}.
\end{theorem}

\subsection{Reducing Reachability to Synthesis}\label{r2s}

Here, we present the reduction from reachability to synthesis.
The reduction also uses a particular gadget to construct a synthesis instance that can be solved iff the reachability instance can be determined.

\paragraph{Reduction:}
Let $P$ be a program, $L$ be a location in $P$, and $V = \{v_1, \dots, v_n\}$ be
global variables never directly assigned in $P$. We construct
$\mathsf{GadgetR2S}(P,L,V)$, which returns a template program $Q$
with template parameters $S$ and a test suite $T$, as follows: 

\begin{enumerate}

\item For every variable $v_i$, define a fresh template variable
$\boxed{c_i}$.  Let the set of template parameters $S$ be the set
containing each $\boxed{c_i}$. 

\item For every function $p \in P$, define a derived function $p_Q \in Q$.
Replace each function call to $p$ with the corresponding call to $p_Q$. 
Replace each use of a variable $v_i$ with a read from the
corresponding template parameter $\boxed{c_i}$; remove all declarations
of variables $v_i$. 

\item Raise a unique exception
 \lt{REACHED}, at the location in $Q$ corresponding to the location $L$ in
 $P$. As usual, when an exception is raised, control immediately jumps to
 the most recently-executed \emph{try-catch} block matching that exception.
 The exception \lt{REACHED} will be caught iff the location in
 $Q$ corresponding to $L \in P$ would be reached.

\item Define a starting function main$_Q$ that has no inputs and returns an integer value.  
Let main$_P{_Q}$ be the function in $Q$ corresponding to the starting function main$_P$ in $P$. 
\begin{itemize}
\item Insert \emph{try-catch} construct that calls $p_Q$  and returns the value $1$ if the exception \lt{REACHED} is caught.
\item At the end of main$_Q$, return the value $0$.
\item Thus, main$_Q$ has the form 
\begin{lstlisting}[numbers=none,xleftmargin=1.cm,basicstyle=\scriptsize,language=C,mathescape]
int main$_Q$() {
  try {
    main$_P{_Q}$();
  } catch (REACHED) { 
    return $1$;
  } 
  return $0$;
}
\end{lstlisting}
\end{itemize}

\item The derived program $Q$ consists of the finite set of template
parameters $S=\{\boxed{c_1}), \dots, \boxed{c_n}\}$ (Step 1), functions $p_Q$'s (Step 2), and the starting function main$_Q$ (Step 4).

\item The test suite $T$ for $Q$ consists of exactly one test case $Q() =
1$, indicating the case when the exception \lt{REACHED} is raised and
caught.
\end{enumerate}

\paragraph{Example:}  Figure~\ref{test-input-syn} illustrates the reduction using the example from Figure~\ref{test-input}.
The synthesized program can be satisfied by $c_0=-20,c_1=-40$,
corresponding to the input ($x=-20,y=-40$) which reaches $L$ in Figure~\ref{test-input}.

The exception \lt{REACHED} represents a unique signal to main$_Q$ that the location $L$ has been reached.
Many modern languages support exceptions for handling special events,
but they are not strictly necessary for the reduction to succeed. 
Other (potentially language-dependent) implementation techniques could
also be employed.  Or, we could use a tuple to represent the signal, e.g., returning
$(v, \mathsf{false})$ from a function that normally returns $v$ if the
location corresponding $L$ has not been reached and $(1, \mathsf{true})$ as
soon as it has. 
BLAST~\cite{beyer2007software}, a model checker for C programs (which do
not support exceptions), uses \emph{goto} and labels to indicate when a
desired location has been reached. 

\paragraph{Reduction Correctness and Complexity:}
The correctness of the $\mathsf{GadgetS2R}$, which transforms reachability to synthesis, depends on two key invariants\footnote{The full proof is given in the Appendix of~\cite{equiv_unl_tr}.}. 
First, for any $c_i$, execution in the derived template program $Q$ with $\boxed{c_i} \mapsto c_i$ mirrors execution in $P$ with $v_i \mapsto c_i$
up to the point when $L$ is reached (if ever). Second, the exception \lt{REACHED} is raised in $Q$ iff location $L$ is reachable in $P$. 

The complexity of $\mathsf{GadgetR2S}$ is \emph{linear} in the input
instance
$P,L,v_i$. The constructed program $Q$ consists of all functions in $P$
and a starting function main$Q$ having $n$ template variables,
where $n = |\{v_i\}|$. 

This reduction directly leads to the main result for this direction of the equivalence:
\begin{theorem}\label{thrs-simple}
The reachability problem in Definition~\ref{defpr} is
reducible to the template-based synthesis problem in Definition~\ref{defps}.
\end{theorem}

\subsection{Synthesis $\equiv$ Reachability}\label{sec:equiv}

Together, the above two theorems establish the equivalence between the
reachability problem in program verification and the template-based program synthesis.

\begin{corollary}\label{equiv}
The reachability problem in Definition~\ref{defpr} and the template-based synthesis problem in Definition~\ref{defps} are linear-time reducible to each other. 
\end{corollary}

This equivalence is perhaps unsurprising as researchers have long assumed certain relations between program synthesis and verification (e.g., see Section~\ref{sec:related}).
However, we believe that a proof of the equivalence is valuable. 
First, our proof, although straightforward, formally shows that both problems inhabit the same complexity class (e.g., the restricted formulation of synthesis in Definition~\ref{defps} is as hard as the reachability problem in Definition~\ref{defpr}).
Second, although both problems are undecidable in the general case, the linear-time transformations allow existing approximations and ideas developed for one problem to apply to the other one.
Third, in term of practicality, the equivalence allows for direct application of off-the-shelf reachability and verification tools to synthesize and repair programs.  
Our approach is not so different from verification works that transform the interested problems into SAT/SMT formulas to be solved by existing efficient solvers.
Finally, this work can be extended to more complex classes of synthesis and repair problems.
While we demonstrate the approach using linear templates, more general templates can be handled. 
For example, combinations of nonlinear polynomials can be considered using a priority subset of terms (e.g., $t_1=x^2, t_2=xy$, as demonstrated in nonlinear invariant generation~\cite{vuicse2012}).

We hope that these results help raise fruitful cross-fertilization among program verification and synthesis fields that are usually treated separately.
Because our reductions produce reachability problem instances that are rarely encountered by current verification techniques  (e.g., with large guards), they may help refine existing tools or motivate optimizations in new directions.
As an example, our bug repair prototype CETI (discussed in the next Section) has produced reachability instances that hit a crashing bug in KLEE that was confirmed to be important by the developers\footnote{http://mailman.ic.ac.uk/pipermail/klee-dev/2016-February/001278.html}.
These hard instances might be used to evaluate and improve verification and synthesis tools (similar to benchmarks used in annual SAT\footnote{SAT Competitions: http://www.satcompetition.org} and SMT\footnote{SMT competitions: http://smtcomp.sourceforge.net/2016} competitions).

\section{Program Repair using Test-Input Generation}
\label{sec:ceti}

We use the equivalence to develop CETI (Correcting Errors using Test Inputs), a tool for automated program repair (a synthesis problem) using test-input generation techniques (which solves reachability problems). 
We define problem of program repair in terms of template-based program synthesis: 

\begin{definition}
  {\bf Program Repair Problem.}
  Given a program $Q$ that fails at least one test in a finite test suite $T$ and a finite set of parameterized templates $S$, does there exist a set of statements $\{s_i\} \subseteq Q$ and parameter values $c_1, \dots, c_n$ for the templates in $S$ such that $s_i$ can be replaced with $S[c_1, \dots, c_n]$ and the resulting program passes all tests in $T$? 
\end{definition}

This repair problem thus allows edits to multiple program statements (e.g., we can replace both lines 4 and 10 in Figure~\ref{fig1} with parameterized templates).
The \emph{single-edit} repair problem restricts the edits to one statement.

CETI implements the key ideas from Theorem~\ref{thsr-simple} in Section~\ref{s2r} to transform this repair problem into a reachability task solvable by existing verification tools.
Given a test suite and a buggy program that fails some test in the suite, CETI employs the statistical fault localization technique Tarantula~\cite{jones2005empirical} to identify particular code regions for synthesis, i.e., program statements likely related to the defect.
Next, for each suspicious statement and synthesis template, CETI transforms the buggy program, the test suite, the statement and the template into a new program containing a location reachable only when the original program can be repaired.
Thus, by default CETI considers single-edit repairs, but it can be modified to repair multiple lines by using $k$ top-ranked suspicious statements (cf. Angelix~\cite{mechtaev2016angelix}). 
Such an approach increases the search space and thus the computational burden placed on the reachability solver.  

Our current implementation employs CIL~\cite{necula02cil} to parse and modify C programs using repair templates similar to those given in~\cite{konighoferautomated,nguyen2013semfix}.
These templates allow modifying constants, expressions (such as the linear template shown in Section~\ref{sec:motiv}), and logical, comparisons, and arithmetic operators (such as changing $||$ to $\&\&$, $\le$ to $ < $, or $+$ to $-$).
Finally, we send the transformed program to the test-input generation tool KLEE, which produces test values that can reach the designated location. 
Such test input values, when combined with the synthesis template and the suspicious statement, correspond exactly to a patch that repairs the bug.
CETI synthesizes correct-by-construction repairs, i.e., the repair, if found, is guaranteed to pass the test suite.

\subsection{Evaluation}
\label{sectexps}
To evaluate CETI, we use the {\tt Tcas} program from the SIR benchmark~\cite{SIR}.
The program, which implements an aircraft traffic collision avoidance system, has 180 lines of code and 12 integer inputs.
The program comes with a test suite of about 1608 tests and 41 faulty functions, consisting of seeded defects such as changed operators, incorrect constant values, missing code, and incorrect control flow. 
Among the programs in SIR, {\tt Tcas} has the most introduced defects (41), and it has been used to benchmark modern bug repair techniques~\cite{debroy2010using,konighofer2013repair,nguyen2013semfix}.

We manually modify {\tt Tcas}, which normally prints its result on the screen, to instead return its output to its caller, e.g., \lt{printf(\"output is \%d\\n\",v)} becomes {\tt return v}.
For efficiency, many repair techniques initially consider a smaller number of tests in the suite and then verify candidate repairs on the entire suite~\cite{nguyen2013semfix}. 
In contrast, we use all available tests at all times to guarantee that any repair found by CETI is correct with respect to the test suite. 
We find that modern tools such as KLEE can handle the complex conditionals that encode such information efficiently and generate the desired solutions within seconds. 

The behavior of CETI is controlled by customizable parameters.
For the experiments described here, we consider the top $n=80$ from the ranked list of suspicious statements and, then apply the predefined templates to these statements.
For efficiency, we restrict synthesis parameters to be within certain value ranges: constant coefficients are confined to the integral range $[-100000, 100000]$ while the variable coefficients are drawn from the set $\{-1,0,1\}$.

\subsubsection*{Results.}
Table~\ref{results} shows the results with 41 buggy {\tt Tcas} versions.
These experiments were performed on a 32-core 2.60GHz Intel Linux system with 128 GB of RAM.
Column {\bfseries Bug Type} describes the type of defect.
\emph{Incorrect Const} denotes a defect involving the use of the wrong constant, e.g., 700 instead of 600.
\emph{Incorrect Op} denotes a defect that uses the wrong operator for arithmetic, comparison, or logical calculations, e.g., $\ge$ instead of $>$.
\emph{Missing code} denotes defects that entirely lack an expression or statement, e.g., $a\&\&b$ instead of $a\&\&b||c$ or {\tt return a} instead of {\tt return a+b}.
\emph{Multiple} denotes defects caused by several actions such as missing code at a location and using an incorrect operator at another location.
Column {\bfseries T(s)} shows the time taken in seconds. 
Column {\bfseries R-Prog} lists the number of reachability program instances that were generated and processed by KLEE.
Column {\bfseries Repair?} indicates whether a repair was found.

\begin{table}[t]
  \caption{Repair Results for 41 \textsf{Tcas} Defects}
  \centering
  {\footnotesize
    \begin{tabular}{r|l|c|c|c||r|l|c|c|c}
      &\textbf{Bug Type} & \textbf{R-Progs} &\textbf{T(s)}&\textbf{Repair?}
      &&\textbf{Bug Type} & \textbf{R-Progs} &\textbf{T(s)}&\textbf{Repair?}\\
      \hline
       v1&incorrect op           &6143&21 &\checkmark&     v22&missing code      	 &5553&175&--\\            
       v2&missing code		 &6993&27 &\checkmark&     v23&missing code      	 &5824&164&--\\            
       v3&incorrect op      	 &8006&18 &\checkmark&     v24&missing code      	 &6050&231&--\\            
       v4&incorrect op		 &5900&27 &\checkmark&     v25&incorrect op		 &5983&19 &\checkmark\\    
       v5&missing code		 &8440&394&--        &     v26&missing code		 &8004&195&--\\            
       v6&incorrect op		 &5872&19 &\checkmark&     v27&missing code		 &8440&270&--\\            
       v7&incorrect const	 &7302&18 &\checkmark&     v28&incorrect op 		 &9072&11 &\checkmark\\    
       v8&incorrect const	 &6013&19 &\checkmark&     v29&missing code		 &6914&195&--\\            
       v9&incorrect op 		 &5938&24 &\checkmark&     v30&missing code		 &6533&170&--\\            
      v10&incorrect op 		 &7154&18 &\checkmark&     v31&multiple	         &4302&16 &\checkmark\\    
      v11&multiple	         &6308&123&--        &     v32&multiple	         &4493&17 &\checkmark\\    
      v12&incorrect op 		 &8442&25 &\checkmark&     v33&multiple               &9070&224&--\\            
      v13&incorrect const 	 &7845&21 &\checkmark&     v34&incorrect op           &8442&75 &\checkmark\\    
      v14&incorrect const 	 &1252&22 &\checkmark&     v35&multiple    	         &9070&184&--\\            
      v15&multiple	         &7760&258&--        &     v36&incorrect const	 &6334&10 &\checkmark\\    
      v16&incorrect const 	 &5470&19 &\checkmark&     v37&missing code           &7523&174&--\\            
      v17&incorrect const 	 &7302&12 &\checkmark&     v38&missing code  	 &7685&209&--\\            
      v18&incorrect const 	 &7383&18 &\checkmark&     v39&incorrect op		 &5983&20 &\checkmark\\    
      v19&incorrect const 	 &6920&19 &\checkmark&     v40&missing code      	 &7364&136&--\\            
      v20&incorrect op		 &5938&19 &\checkmark&     v41&missing code      	 &5899&29 &\checkmark\\
      v21&missing code      	 &5939&31 &\checkmark&        &&&
    \end{tabular}
  }
  \label{results}
\end{table}

We were able to correct 26 of 41 defects, including multiple defects of different types. On average, CETI takes 22 seconds for each successful repair. 
The tool found 100\% of repairs for which the required changes are single edits according to one of our predefined templates (e.g., generating arbitrary integer constants or changing operators at one location). 
In several cases, defects could be repaired in several ways.
For example, defect $v_{28}$ can be repaired by swapping the results of both branches of a conditional statement or by inverting the conditional guard. 
CETI also obtained unexpected repairs. 
For example, the bug in $v_{13}$ is a comparison against an incorrect constant; the buggy code reads {\tt < 700} while the human-written patch reads {\tt < 600}. 
Our generated repair of {\tt < 596} also passes all tests. 

We were not able to repair 15 of 41 defects, each of which requires edits at multiple locations or the addition of code that is beyond the scope of the current set of templates.  
As expected, CETI takes longer for these programs because it tries all generated template programs before giving up. 
One common pattern among these programs is that the bug occurs in a macro definition, e.g., {\tt \#define C = 100} instead of {\tt \#define C = 200}. 
Since the CIL front end automatically expands such macros, CETI would need to individually fix each use of the macro in order to succeed.
This is an artifact of CIL, rather than a weakness inherent in our algorithm. 

CETI, which repairs 26 of 41 {\tt Tcas} defects, performs well compared to other reported results from repair tools on this benchmark program. 
GenProg, which finds edits by recombining existing code, can repair 11 of these defects~\cite[Tab.~5]{nguyen2013semfix}.
The technique of Debroy and Wong, which uses random mutation, can repair 9 defects~\cite[Tab.~2]{debroy2010using}.
FoREnSiC, which uses the concolic execution in CREST, repairs 23 defects~\cite[Tab.~1]{konighofer2013repair}.
SemFix out-performs CETI, repairing 34 defects~\cite[Tab.~5]{nguyen2013semfix}, but also uses fifty test cases instead of the entire suite of thousands\footnote{Thus CETI's repairs, which pass the entire suite instead of just 50 selected tests, meet a higher standard. We were unable to obtain SemFix details, e.g., which 50 tests, online or from the authors.}.
Other repair techniques, including equivalence checking~\cite{konighofer2013repair} and counterexample guided refinement~\cite{konighofer2013repair}, repair 15 and 16 defects, respectively. 

Although CETI uses similar repair templates as both SemFix and FoREnSiC, the repair processes are different.
SemFix directly uses and customizes the KLEE symbolic execution engine, and FoRenSiC integrates concolic execution to analyze programs and SMT solving to generate repairs.
In contrast, CETI eschews heavyweight analyses, and it simply generates a reachability instance. 
Indeed, our work is inspired by, and generalizes, these works, observing that the whole synthesis task can be offloaded with strong success in practice.

However, there is a trade-off: customizing a reachability solver to the task of program repair may increase the performance or the number of repairs found, but may also reduce the generality or ease-of-adoption of the overall technique. 
We note that our unoptimized tool CETI already outperforms published results for GenProg, Debroy and Wong, and FoREnSiC on this benchmark, and is competitive with SemFix. 

\subsubsection*{Limitations.}

We require that the program behaves deterministically on the test cases and that the defect be reproducible. 
This limitation can be mitigated by running the test cases multiple times, but ultimately our technique is not applicable if the program is non-deterministic.
We assume that the test cases encode all relevant program requirements. 
If adequate test cases are not available then the repair may not retain required functionality.
Our formulation also encodes the test cases as inputs to a starting function (e.g., {\tt main}) with a single expected output.
This might not be feasible for certain types of specifications, such as liveness properties (``eventually'' and ``always'') in temporal logic.
The efficiency of CETI depends on fault localization to reduce the search space.
The reachability or test-input generation tool used affects both the efficiency and the efficacy of CETI. 
For example, if the reachability tool uses a constraint solver that does not support data types such as string or arrays then we will not be able to repair program defects involving those types. 
Finally, we assume that the repair can be constructed from the provided repair templates. 

The reduction in Section~\ref{s2r} can transform a finite space (buggy) program into an infinite space reachability problem (e.g., we hypothesize that a bounded loop guard $i \le 10$ is buggy and try to synthesize a new guard using an unknown parameter $i \le \BOX{c}$).
However, this does not invalidate the theoretical or empirical results and the reduction is efficient in the program size and the number of tests.
The reduction also might not be optimal if we use complex repair templates (e.g., involving many unknown parameters).
In practice we do not need to synthesize many complex values for most defects and thus modern verification tools such as KLEE can solve these problems efficiently, as shown in our evaluation.

This paper concretely demonstrates the applicability of program reachability (test-input generation) to program synthesis (defect repair) but not the reverse direction of using program synthesis to solve reachability. 
Applying advances in automatic program repair to find test-inputs to reach nontrivial program locations remains future work.

\section{Related Work}\label{sec:related}
\subsubsection*{Program Synthesis and Verification.} 
Researchers have long hypothesized about the relation between program synthesis and verification and proposed synthesis approaches using techniques or tools often used to verify programs such as constraint solving or model checking~\cite{alur2015syntax,Srivastava:2010:PVP:1707801.1706337}.
For example, Bodik and Solar-Lezama et. al.'s work~\cite{solar-pldi07,lezama2008program} on sketching defines the synthesis task as: $\exists c~.~\forall(i,o)~.~ \in T ~.~ (P[c])(i) = o$ (similar to our template-based synthesis formulation in Definition~\ref{defps}) and solves the problem using a SAT solver.
Other synthesis and program repair researches, e.g.,~\cite{attie2015model,mechtaev2016angelix,nguyen2013semfix,Srivastava:2010:PVP:1707801.1706337,srivastava2013template}, also use similar formulation to integrate verification tools, e.g., test-input generation, to synthesize desired programs.
In general, such integrations are common in many ongoing synthesis works including the multi-disciplinary ExCAPE project~\cite{excape} and the SyGuS competition~\cite{sygus}, and have produced many practical and useful tools such as Sketch that generates low-level bit-stream programs~\cite{lezama2008program}, Autograder that provides feedback on programming homework~\cite{singh2013automated}, and FlashFill that constructs Excel macros~\cite{flashfill,flashfill2}.

The work presented in this paper is inspired by these works, and generalizes them by establishing a formal connection between synthesis and verification using the template-based synthesis and reachability formulations. 
We show that it is not just a coincident that the aforementioned synthesis works can exploit verification techniques, but that every template-based synthesis problem can be reduced to the reachability formulation in verification.
Dually, we show the other direction that reduces reachability to template-based synthesis, so that every reachability problem can be solved using synthesis.
Furthermore, our constructive proofs describe efficient algorithms to do such reductions.

\subsubsection*{Program Repair and Test-Input Generation.}
Due to the pressing demand for reliable software, automatic program repair has steadily gained research interests and produced many novel repair techniques.
\emph{Constraint-based} repair approaches, e.g., AFix~\cite{Jin:2011:AAF:1993498.1993544}, Angelix~\cite{mechtaev2016angelix}, SemFix~\cite{nguyen2013semfix}, FoRenSiC~\cite{bloem2013forensic}, Gopinath et al.~\cite{gopinath2011specification}, Jobstmann et al.~\cite{jobstmann2005program}, generate constraints and solve them for patches that are correct by construction (i.e., guaranteed to adhere to a specification or pass a test suite).
In contrast, \emph{generate-and-validate} repair approaches, e.g., GenProg~\cite{icse09}, Pachika~\cite{1747554}, PAR~\cite{kim2013automatic}, Debroy and Wong~\cite{debroy2010using}, Prophet~\cite{long2016automatic}, find multiple repair candidates (e.g., using stochastic search or invariant inferences) and verify them against given specifications.

The field of test-input generation has produced many practical techniques and tools to generate high coverage test data for complex software, e.g., fuzz testing~\cite{fuzz,forrester2000empirical}, symbolic execution~\cite{cadar2008klee,cadar2013symbolic}, concolic (combination of static and dynamic analyses) execution~\cite{godefroid2005dart,sen2006cute}, and software model checking~\cite{beyer2007software,ball2002s}.
Companies and industrial research labs such as Microsoft, NASA, IBM, and Fujitsu have also developed test-input generation tools to test their own products~\cite{anand2007jpf,artzi2008finding,godefroid2008automated,li2011klover}.
Our work allows program repair and synthesis approaches directly apply these techniques and tools.

\section{Conclusion}
We constructively prove that the template-based program synthesis problem and the reachability problem in program verification are equivalent.
This equivalence connects the two problems and enables the application of ideas, optimizations, and tools developed for one problem to the other.
To demonstrate this, we develop CETI, a tool for automated program repair using test-input generation techniques that solve reachability problems. 
CETI transforms the task of synthesizing program repairs to a
reachability problem, where the results produced by a test-input
generation tool correspond to a patch that repairs the original program.
Experimental case studies suggest that CETI has higher success rates than many other standard repair approaches.

\subsubsection*{Acknowledgments.} 
This research was partially supported by NSF awards CCF 1248069, CNS 1619098, CNS 1619123, as well as AFOSR grant FA8750-11-2-0039 and DARPA grant FA8650-10-C-7089.

\newpage
\bibliographystyle{abbrv}
\bibliography{bibs.bib}

\ifproofs

\newpage
\section{Appendix: Correctness Proofs}

\subsection{Formal Semantics}\label{semantics}

We use standard operational semantics to reason about the meanings and
executions of programs~\cite{plotkin04}. A state $\sigma$ maps 
variables to values. Given a program $P$, the \emph{large-step} judgment 
$\langle P, \sigma \rangle \Downarrow \sigma'$ means that $P$, when
executed in an initial state $\sigma$, terminates in state $\sigma'$. 
The \emph{small-step} judgment $\langle P, \sigma \rangle \rightarrow
\langle P', \sigma' \rangle$ means that if $P$ is executed in state
$\sigma$ for one single step, it transitions to (or reduces to) program $P'$
in state $\sigma'$. We write $\rightarrow^*$ for the reflexive, transitive
closure of that relation (i.e., zero or more steps). Both large-step and
small-step judgments are proved or constructed using
\emph{derivations}, tree-structured mathematical objects made of
well-formed instances of inference rules. We write $D :: \dots$ 
to indicate that $D$ is a derivation proving a judgment. 
There is often only one inference rule for each statement type (e.g.,
one assignment rule for reasoning about assignment statements), which
admits reasoning by \emph{inversion}: for example, if $D ::
\langle P, \sigma \rangle \Downarrow \sigma'$ and $P$ is an assignment
statement, then $D$ must be the assignment rule (and not, for example, the
while loop rule), and vice-versa. 

We write $\sigma \models B$ to denote that a state $\sigma$ makes a
predicate $B$ true. We write $\mathit{wp}(P, B)$ for the weakest
precondition of $P$ with respect to postcondition $B$.  If $A =
\mathit{wp}(P, B)$, then if $P$ executes in a state that makes the
precondition $A$ true and it terminates, then it terminates in a state
making postcondition $B$ true. 

We use large-step operational semantics when reasoning about synthesis,
where we focus on the final value of the program (i.e., its behavior on a
test).  We use small-step operational semantics when reasoning about
reachability, where intermediate steps matter (i.e., did
the program execution visit a particular label?). We use induction on the structure of a
derivation to show that a property holds for all executions of all programs
(informally, if the property holds for one-step programs like
\lt{skip}, and it also holds whenever one more execution step is added,
then it always holds). We use weakest preconditions to reason about 
special conditional statements that encode test cases. 

Our use of small-step and large-step semantics, as well as structural 
induction on derivations and preconditions, is standard (i.e., the proofs
and the constructions are novel, but not the proof machinery). The reader is
referred to Plotkin~\cite{plotkin04} for a thorough introduction.

\subsection{Reducing Synthesis to Reachability}\label{s2r_proof}

We prove correct the constructive reduction given in Section~\ref{s2r}.
The correctness of $\mathsf{GadgetS2R}$ hinges on two key invariants. 
First, function calls in the derived program $P$ 
have the same behavior as template functions in the original program $Q$.
Second, location $L$ is reachable if and only if values $c_i$ can be
assigned to variables $v_i$ 
such that $Q$ passes all of the tests. Formally, we say that 
$\mathsf{GadgetS2R}$ maintains the invariants described by the following
Lemmas. 
For brevity, we write
$\sigma_1 \simeq_V \sigma_2$ to denote
$\forall x .\;x \not \in V \Rightarrow \sigma_1(x) = \sigma_2(x)$ --- that is,
state equivalence modulo a set of variables $V$. 

We first show that $p_q(\mathrm{input})$ behaves as
$q[c_1, \dots, c_n](\mathrm{input})$, when the new $v_i$
variables in $P$ are assigned the values $c_i$. 
Formally, 
\begin{lemma}\label{lemma-execution} 
Let $(P, L) = \mathsf{GadgetS2R}(Q,S,T)$. 
For all states $\sigma_1, \sigma_2, \sigma_3$, 
all values $c_1, \dots, c_n, i$, and all functions $q \in Q$,
if 
$\sigma_1(v_i) = c_i$  
then 
$D_1 :: \langle p_q(i), \sigma_1 \rangle \Downarrow \sigma_2$ 
if both
$D_2 :: \langle q[c_1, \dots, c_n](i), \sigma_1 \rangle \Downarrow \sigma_3$
and 
$\sigma_2 \simeq_{\{v_1,\dots,v_n\}} \sigma_3$. 
\end{lemma} 

\begin{proof}
The proof proceeds by induction on
the structure of the operational semantics derivation of $D_1$. 
Let $\sigma_1$ be arbitrary with $\sigma_1(v_i) = c_i$. 

Note that by Step 2 of the $\mathsf{GadgetS2R}$ construction, each
$p_q(i) \in P$ corresponds to $q[c_1, \dots, c_n](i) \in
Q$ in a particular manner. Indeed, all subexpressions in $p_q$ and $q$ are
identical \emph{except} for one case. References to template
parameters $\boxed{c_i}$ in the template program $Q$ correspond exactly to
references to variables $v_i$ in the derived reachability program
$P$. 

Thus, by inversion, the structure of $D_1$ corresponds exactly to the
structure of any $D_2$
except for variable references.  
The inductive proof considers all of the cases 
for the derivation
$D_1$. For brevity, we show only the non-trivial case.

{\sc Case 1 (Template Variable Read)}. Suppose $D_1$ is:
$$
\infer[\textsf{assign-variable}]
{
\langle a := v_i, \sigma_1 \rangle \Downarrow \sigma_2
}
{
\sigma_2 = \sigma_1[a \mapsto \sigma_1(v_i)]
} 
$$
By inversion and the construction of $P$, $D_2$ is: 
$$
\infer[\textsf{assign-template-parameter}]
{
\langle a := \boxed{c_i}, \sigma_1 \rangle \Downarrow \sigma_3
}
{
\sigma_3 = \sigma_1[a \mapsto c_i]
}
$$
Note that in $D_1$, $v_i$ is a normal program variable expression
(referring to a variable introduced in Step 1 of the construction), 
while in in $D_2$, $\boxed{c_i}$ is a reference to a template parameter. We
must show that
$\sigma_2 \simeq_{\{v_1,\dots,v_n\}} \sigma_3$. 
Since $\sigma_2$ and $\sigma_3$ agree with $\sigma_1$ on all variables
except $a$, it only remains to show that
$\sigma_2(a) = \sigma_3(a)$ (since $a \not \in \{v_1,\dots,v_n\}$). 
By the $\sigma$ definitions from $D_1$ and $D_2$, that obligation
simplifies to showing that $\sigma_1(v_i) = c_i$, which was part of
the input formulation for this Lemma. (Intuitively, 
$\sigma_1(v_i) = c_i$ means that the reachability analysis assigned
the values $c_i$ to each variable $v_i$.)

The other cases, which do not involve the template parameters $v_i$, 
are direct (the inference rule used in $D_1$ will exactly mirror the 
inference rule used in $D_2$). 
\end{proof} 

Having established that the executions of $P$'s functions mirror
$Q$'s functions (modulo the
template parameters, which are held constant), we now establish that
reaching $L$ in $P$ via assigning each $v_i$ the value $c_i$ 
corresponds to 
$Q[c_1, \dots, c_n]$ passing all of the tests. 

\begin{lemma}\label{lemma-wp} 
Let $(P, L) = \mathsf{GadgetS2R}(Q,S,T)$. 
The execution of $P$ reaches $L$ starting
from state $\sigma_1$ if and only if  
$ \sigma_1 \models 
  wp(Q[c_1, \dots, c_n](i_1), \mathrm{result} = o_1)
  \wedge 
  \dots
  \wedge 
  wp(Q[c_1, \dots, c_n](i_m), \\ \mathrm{result} = o_m) 
$, where
$\sigma_1(v_i) = c_i$, 
$\mathrm{result}$ denotes the return value of $Q$, 
and $wp$ is the weakest precondition. 
\end{lemma} 

\begin{proof}
By construction, there is exactly one location $L$ in $P$: 
``if $e$: [L]'' is the body of $p_\text{gadget} \in P$. 
From Step 3,
$e$ has the general form $f(i_1) = o_1 \wedge \dots 
\wedge f(i_m) = o_m$.  
By 
the standard weakest precondition definitions for if,
conjunction, equality and function calls, we have
that $L$ is reachable if and only if 
$\sigma_1 \models
  wp(p_\textrm{gadget}(i_i), 
  \mathrm{result} = o_1)
  \wedge 
  \dots
  \wedge 
  wp(p_\textrm{gadget}(i_m), \mathrm{result} = o_m)$. 
That is, the label is reachable iff the derived gadget program passes the
tests (with variables $v_i \mapsto c_i$). We now show that this occurs if
and only if the original template program passes the tests (with template
parameters $\boxed{c_i} \mapsto c_i$). 

By conjunction elimination, for each test $(i,o)$ we have 
$\sigma_1 \models
  wp(p_\textrm{gadget}(i), \mathrm{result} = o)$. 
By the soundness and completeness of weakest preconditions with respect to
operational 
semantics,\footnote{Our approach is thus sound and relatively
complete in practice. An implementation would use provability ($\vdash$),
such as from an SMT solver, instead of truth ($\models$). While sometimes
sound on restricted domains, decision procedures are not complete in
general, in which case a failure to solve the reachability problem $P$ does
\emph{not} imply that the synthesis problem $Q$ has no solution.} 
we have that 
$\langle p_\textrm{gadget}(i), \sigma_1 \rangle
  \Downarrow \sigma_2 $ iff $\sigma_2 \models \mathrm{result} = o$
  (equivalently, $\sigma_2(\mathrm{result}) = o$). 
By Lemma~\ref{lemma-execution}, we have 
$\langle Q[c_1, \dots, c_n](i), \sigma_1 \rangle
  \Downarrow \sigma_3 $ with state 
$\sigma_2 \simeq_{\{v_1,\dots,v_n\}} \sigma_3$. Since
$\mathrm{result} \not \in \{v_1, \dots, v_n\}$, we have
$\sigma_3(\mathrm{result}) = \sigma_2(\mathrm{result}) = o$. 
That is, the template program $Q$, filled in with the concrete values $c_i$,
when run on any test input $i$ yields the expected result $o$. This occurs
if and only if an execution of $P$ reaches $L$ starting
in a state that maps each $v_i$ to $c_i$. 
\end{proof} 

This Lemma leads directly to the main result for this direction of the
reduction. 

\begin{theorem}\label{thsr}
Let $Q$ be a program with template parameters
$S=\{\boxed{c_1}, \dots, \boxed{c_n}\}$ and a test suite $T=\{
(i_1,o_1), \dots \}$. 
Let $(P, L) = \mathsf{GadgetS2R}(Q,S,T)$. 
Then there exist parameter values $c_i$ such that $\forall (i,o)
\in T ~.~ (Q[c_1, \dots, c_n])(i) = o$
if and only if 
there exist input values $t_i$ such that 
the execution of $P$ with $v_i \mapsto t_i$ reaches $L$. 
That is, the template-based synthesis problem in Definition~\ref{defps} is reducible to the reachability problem in
Definition~\ref{defpr}.
\end{theorem}

\begin{proof}
This follows directly from Lemma~\ref{lemma-wp} with $t_i$ equal 
to $c_i$. 
\end{proof}

\subsection{Reducing Reachability to Synthesis}\label{r2s_proof}
We now address the correctness of the constructive reduction given in Section~\ref{r2s}.
The correctness of $\mathsf{GadgetR2S}$ also relies on two key invariants. 
First, for any $c_i$, execution in the derived templated
program $Q$ with $\boxed{c_i} \mapsto c_i$ mirrors execution in $P$ with
$v_i \mapsto c_i$
up to the point when $L$ is reached (if ever).
Second, the exception \lt{REACHED} is raised
in $Q$ iff location $L$ is reachable in $P$. 

We first show that $q_p[c_1, \dots,
c_n]$ behaves as $p$ with $v_i \mapsto c_i$. Because our construction uses
exceptions, we phrase the lemma formally using small-step operational
semantics. 

\begin{lemma} \label{lemma-3}
Let $Q, S, T = \mathsf{GadgetR2S}(P, L, v_i)$. For all states 
$\sigma_1, \sigma_2, \sigma_3$, all values $c_i$, 
all functions $p \in P$, and all programs $p'$ and $q_p'$, 
if $\sigma_1(v_i) = c_i$ then 
we have
$D_1 :: \langle p, \sigma_1 \rangle \rightarrow^*
        \langle p', \sigma_2 \rangle$
with $L$ not executed in $D_1$
iff  
$D_2 :: \langle q_p, \sigma_1 \rangle \rightarrow^*
        \langle q_p', \sigma_3 \rangle$ 
with \lt{raise REACHED} not executed in $D_2$
and $\sigma_2 \simeq_{v_i} \sigma_3$ where 
$q_p'$ is equal to $p'$ with each $v_i$ replaced
by $\boxed{c_i}$ and each $[L]$ replaced by
\lt{raise REACHED}. 
\end{lemma} 

\begin{proof}
The proof proceeds by induction on the structure of the operational semantics derivation $D_1$.  The proof follows the reasonings shown in Section~\ref{s2r} and is elided in the interest of brevity.
\end{proof}

Having established that the execution of $Q$'s functions mirrors 
the execution of $P$'s functions before
the location $L$ is reached (when $\boxed{c_i} = v_i = c_i$),
we now demonstrate that exception \lt{REACHED} is raised
in $Q$ iff location $L$ is reachable in $P$:

\begin{lemma} \label{lemma-4}
Let $Q, S, T = \mathsf{GadgetR2S}(P, L, v_i)$. For all values $c_i$,
$Q[c_1, \dots, c_n]() = 1$ 
iff location $L$ is reachable in $P$ starting from a state with
$v_i \mapsto c_i$. 
\end{lemma} 

\begin{proof}
The proof follows from the instantiation of Lemma~\ref{lemma-3}. 
By construction and the operational semantics rules for \emph{try-catch},
$q_{main}$ returns $1$ iff $q_p$ did raise \lt{REACHED}. 
By construction, there is exactly one occurrence of 
\lt{raise REACHED} in $Q$, and that statement occurs at a point
corresponding to the singular location $L$ in $P$. Note that $q_p$ is derived
from $p$ and shares the same structure except that reads from variables
$v_i$ in $P$ are replaced with reads from template variables $\boxed{c_i}$
in $Q$ and the location $L$ is replaced by \lt{raise REACHED}. 

Thus we apply Lemma~\ref{lemma-3} with $p = p_{main}$, $q_p = q_{main}$
and $\sigma_1$ such that $\sigma_1(v_i) = c_i$ (from the
statement of Lemma~\ref{lemma-4}). Since Lemma~\ref{lemma-3} applies to
all program points $p'$ and $q_p'$ up to the first execution of $L$
(resp. \lt{raise REACHED}), we have a derivation $D_1$ starting in
$\sigma_1$ and ending in $p' = [L]; p''$ iff we have a matching derivation
$D_2$ starting in $\sigma_1$ and ending in $q_p' = $ \lt{raise REACHED}$;
q_p''$. 

Thus, $L$ is reachable in $P$ starting from a state $\sigma_1(v_i) = c_i$ 
iff $Q[c_1, \dots, c_n]()$ returns $1$. 
\end{proof} 

\noindent This Lemma leads directly to the Theorem for this direction of the
reduction.

\begin{theorem}\label{thrs}
Let $P$ be a program with a location $L$ and global variables $v_1, \dots, v_n$.
Let $Q, S, T = \mathsf{GadgetR2S}(P,L,v_i)$. 
Then there exist input values $t_i$ such that 
the execution of $P$ with $v_i \mapsto t_i$ reaches $L$ 
iff there exist parameter values $c_i$ such that
$\forall (i,o) \in T. (Q[c_1, \dots, c_n])(i) = o$. 
That is, 
the reachability problem in Definition~\ref{defpr} is
reducible to the template-based synthesis problem in Definition~\ref{defps}.
\end{theorem}

\begin{proof}
This follows from Lemma~\ref{lemma-4} with $t_i$ equal to $c_i$. Note 
that $T$ is the singleton set containing $(\emptyset, 1)$ by construction,
so the universal qualifier reduces to the single assertion
$Q[c_1, \dots, c_n]() = 1$. 
\end{proof} 
\fi

\end{document}